\documentclass[11pt]{article}

\usepackage{epsfig}
\usepackage{amsfonts}
\usepackage{amssymb}
\usepackage{amstext}
\usepackage{amsmath}
\usepackage{enumerate}
\usepackage{algorithm}
\usepackage{algorithmic}
\usepackage{bbm}

\usepackage{amsthm}

\usepackage[usenames,dvipsnames,svgnames,table]{xcolor}

\usepackage[pagebackref,letterpaper=true,colorlinks=true,pdfpagemode=UseNone,citecolor=OliveGreen,linkcolor=BrickRed,urlcolor=BrickRed,pdfstartview=FitH]{hyperref}


\newtheorem{lemma}{Lemma}
\newtheorem{theorem}{Theorem}

\newtheorem{corollary}{Corollary}

\newtheorem{claim}{Claim}

\numberwithin{equation}{section}
\numberwithin{table}{section}

\newcommand{\junk}[1]{}

\renewcommand{\l}{\lambda}

\newcommand{\vol}{{\rm vol}}

\def\b1{{\bf 1}}
\def\eps{{\epsilon}}

\def\ver{{\phi^V}}

\def\vol{\operatorname{vol}} 
\def\polylog{\operatorname{polylog}} 
\def\expe{\mathbb E} 
\def\prob{\mathbb P} 

\renewcommand{\leq}{\leqslant}
\renewcommand{\geq}{\geqslant}

\newcommand{\E}{\mathop{\mathbb E}}
\renewcommand{\Pr}{\mathbb P}

\newcommand{\R}{\mathbb R}

\newcommand{\abs}[1]{\lvert#1\rvert}

\newcommand{\card}[1]{\lvert#1\rvert}
\newcommand{\norm}[1]{\lVert#1\rVert}

\newcommand{\set}[1]{\{#1\}}

\newcommand{\inner}[1]{\langle#1\rangle}

\newcommand{\domain}[1]{[#1]}

\newcommand{\var}{\mathop{\textnormal{Var}}\nolimits}
\newcommand{\cov}{\mathop{\textnormal{Cov}}\nolimits}
\newcommand{\gap}{1-\lambda}
\newcommand{\cgap}{\varphi}
\newcommand{\mix}{{\rm mix}}
\newcommand{\1}{\mathbbm{1}}

\renewcommand{\le}{\leq}

\title{Random Walks and Evolving Sets:\\
Faster Convergences and Limitations}

\author{
Siu On Chan\thanks{The Chinese University of Hong Kong, \protect\url{sochan@gmail.com}.}
\and
Tsz Chiu Kwok\thanks{\'Ecole polytechnique f\'ed\'erale de Lausanne, \protect\url{tckwok0@gmail.com}.}
\and
Lap Chi Lau\thanks{University of Waterloo, \protect\url{lapchi@uwaterloo.ca}. 
}
}

\date{}

\begin{document}

\begin{titlepage}
\def\thepage{}
\thispagestyle{empty}

\maketitle

\begin{abstract}
\noindent
Analyzing the mixing time of random walks is a well-studied problem with applications in random sampling and more recently in graph partitioning.
In this work, we present new analysis of random walks and evolving sets using more combinatorial graph structures, and show some implications in approximating small-set expansion.
On the other hand, we provide examples showing the limitations of using random walks and evolving sets in disproving the small-set expansion hypothesis.
\begin{enumerate}
\item We define a combinatorial analog of the spectral gap, and use it to prove the convergence of non-lazy random walks.  
A corollary is a tight lower bound on the small-set expansion of graph powers for any graph.
\item We prove that random walks converge faster when the robust vertex expansion of the graph is larger.
This provides an improved analysis of the local graph partitioning algorithm using the evolving set process.
\item We give an example showing that the evolving set process fails to disprove the small-set expansion hypothesis.
This refutes a conjecture of Oveis Gharan and shows the limitations of local graph partitioning algorithms in approximating small-set expansion.
\end{enumerate}
\end{abstract}

\end{titlepage}

\newpage

\section{Introduction}

Analyzing the mixing time of random walks is a fundamental problem with many applications in random sampling~\cite{LPW}.
The evolving set process is an elegant tool introduced by Morris and Peres~\cite{MP} to provide sharp analyses of mixing time (see the survey~\cite{MT}).
Recently, random walks and the evolving set process have also been used in designing local algorithms for graph partitioning~\cite{ST,ACL,AP,OT,KL12}.
The evolving set process is the most powerful among the local graph partitioning algorithms, and Oveis Gharan~\cite{O} even conjectured that it can be used to disprove the small-set expansion hypothesis~\cite{RS}.
A common theme of this paper is to study the power and limitations of this technique, from analyzing mixing time to local graph partitioning and approximating small-set expansion.

\subsubsection*{Random Walks and Mixing Time}

We consider random walks in a weighted undirected graph $G=(V,E)$ with a nonnegative weight $w(e)$ on each edge $e \in E$.
Let $n = |V|$ and $m = |E|$.
For simplicity, we assume that the graph $G$ is regular
and the weights are scaled such that the weighted degree of $i$ is $\sum_{j: ij \in E} w(i,j) = 1$ for all $i \in V$ throughout this paper,
but we will mention how to deal with general graphs in Section~\ref{s:general}.
Let $A$ be the $n \times n$ random walk matrix of $G$ with $A_{ij} = w(i,j)$.
Let $p_0 \in \R^n$ be an initial probability distribution, and let $p_t := A^t p_0$ be the probability distribution after $t$ steps of random walks.
When $G$ is connected and non-bipartite,
it is well-known that $p_t$ will converge to the uniform distribution.
The mixing time is defined as 
\[T_{\mix} = \min\{t: \norm{p_t - \frac{\vec{1}}{n}}_1 \leq 1/4 {\rm~for~all~initial~distribution~} p_0\}.\]

One approach to analyze the mixing time is to look at the eigenvalues of the random walk matrix.
Let the eigenvalues of $A$ be $1 = \lambda_1 \geq \lambda_2 \geq \ldots \geq \lambda_n \geq -1$.
By basic spectral graph theory, it can be shown that $1 > \lambda_2$ if and only if $G$ is connected, and $\lambda_n > -1$ if and only if $G$ is non-bipartite.
This implies that, when $G$ is connected and non-bipartite, $p_t$ will converge to the first eigenvector, and thus the uniform distribution is the unique limiting distribution of the random walk.
Let $\lambda:=\max\{\lambda_2, |\lambda_n|\}$, 
and let $\gap$ be the spectral gap of the random walk matrix.
A standard calculation shows that the mixing time is upper bounded by $O(\log(n)/(\gap))$.

For many problems, it is useful to have combinatorial characterizations of graphs with fast mixing time.
For two subsets $S, T \subseteq V$,
let $E(S,T)$ be the set of edges with one vertex in $S$ and another vertex in $T$, and let $w(S,T) := \sum_{e \in E(S,T)} w(e)$.
The expansion of a set $S \subseteq V$ and the expansion of a graph $G$ are defined as 
\[\phi(S) := \frac{w(S,V-S)}{|S|} \quad {\rm and} \quad \phi(G) := \min_{S:|S| \leq n/2} \phi(S).\]
Cheeger's inequality for graphs~\cite{A,AM} states that 
\[\frac12 (1-\lambda_2) \leq \phi(G) \leq \sqrt{2(1-\lambda_2)},\] 
and thus $1-\lambda_2 = \Omega(\phi(G)^2)$.
Having large conductance is not enough to guarantee fast mixing time, as $\lambda_n$ may be very close to $-1$.
There is a simple trick to bypass this issue: one can guarantee that $\lambda_n \geq 0$ by considering ``lazy'' random walks (with probability $1/2$ stay put),
and this implies that the mixing time of lazy random walks is upper bounded by $O(\log(n)/\phi(G)^2)$.

Another approach to analyze the mixing time is to directly use the graph structures.
Lov\'asz and Simonovits~\cite{LS} developed a combinatorial method to prove that the mixing time of lazy random walks is $O(\log(n)/\phi(G)^2)$.
This method is more flexible in incorporating additional graph structures.
Given a parameter $0 < \delta \leq 1/2$,
the $\delta$-small-set expansion is defined as 
\[\phi_{\delta}(G) := \min_{S:|S| \leq \delta n} \phi(S).\]
Lov\'asz and Kannan~\cite{LK} proved that the mixing time of lazy random walks is 
\[T_{\mix} \leq \int_{1/n}^{1/2} \frac{dx}{x \phi_x(G)^2}.\]
As we will discuss in more details shortly,
this combinatorial approach can also be used to design local graph partitioning algorithms for approximating small-set expansion.

\subsubsection*{Evolving Sets}

The evolving set process is a Markov chain on subsets of $V$ with the following transition rule:
If the current set is $S$,
choose $U$ uniformly from $[0,1]$ and the next set is defined as 
\[\tilde{S} := \{y: w(y,S) \geq U\}.\]
Morris and Peres~\cite{MP} used the evolving set process to strengthen Lov\'asz and Kannan result to bound the uniform mixing time of lazy random walks by the expansion profile. 
An important definition in their analysis is the gauge of a set $S$ and the gauge of a graph $G$, which are defined as 
\[\psi(S) = 1 - \expe[\sqrt{|\tilde S| / |S|}] \quad {\rm and} \quad \psi(G) := \min_{S:|S| \leq n/2} \psi(S).\]
Morris and Peres~\cite{MP} showed that the convergence rate of random walks is bounded by the gauge, and the mixing time of random walks is $O(\log(n) / \psi(G))$.
They proved that $\psi(G) \geq \Omega(\phi^2(G))$ for lazy graphs, and this implies that the mixing time of lazy random walks is upper bounded by $O(\log(n) / \phi(G)^2)$.
We refer the interested reader to~\cite{LPW} for an excellent introduction of the evolving set process.

\subsubsection*{Local Graph Partitioning Algorithms}

Spielman and Teng~\cite{ST} used random walks to design the first local graph partitioning algorithm, which outputs a set $S$ of approximately optimal expansion with running time depends only on $|S|$ and $\polylog(n)$.
Their analysis is based on the approach of Lov\'asz and Simonovits~\cite{LS} on analyzing mixing time.
Andersen and Peres~\cite{AP} and Oveis Gharan and Trevisan~\cite{OT} used the evolving set process for local graph partitioning, and provide the current best known algorithm in terms of both the approximation ratio and the running time.
\begin{theorem}[\cite{AP,OT}] \label{t:approx}
For any target set $S^*$ and any $\eps > 0$, there is a subset $S'$ with $|S'| \geq |S^*|/2$, such that if we start the evolving set process with $\{v\}$ for $v \in S'$,
then with constant probability the algorithm returns a set $S$ with $\phi(S) = O(\sqrt{\phi(S^*)/\eps})$ and $|S| = O(|S^*|^{1+\eps})$, and the running time is $O(|S^*|^{1+2\eps} \phi(S)^{-\frac12} \log^2 n)$.
\end{theorem}

\subsubsection*{Small Set Expansion}

The small set expansion hypothesis proposed by Raghavendra and Steurer~\cite{RS} states that for any $\eps$, there exists $\delta$ such that it is NP-hard to distinguish the following two cases:
\begin{enumerate}
\item There is a set $S$ with $\phi(S) \leq \eps$ and $|S| \leq \delta n$;
\item $\phi(S) \geq 1-\eps$ for every set $S$ with $|S| \leq \delta n$.
\end{enumerate}
This hypothesis is shown to be closely related to the unique games conjecture~\cite{RS}.
The local graph partitioning algorithms provide bicriteria approximation algorithms for computing small set expansion $\phi_{\delta}(G)$.
It is observed in~\cite{OT,KL12} that if the output size guarantee of the above local graph partitioning algorithm is improved from $O(|S^*|^{1+\eps})$ to $O(|S^*|)$, then the small-set expansion hypothesis is false.
Oveis Gharan~\cite{O} suggested a plan to prove such an output size guarantee using the evolving set process.

\subsection{Our Results}

\subsubsection*{Combinatorial Analog of Spectral Gap}

We define a combinatorial analog of spectral gap with which we can directly analyze the mixing time of non-lazy random walks.
Recall that $\phi(G)$ is defined as
\[
\min_{S \subseteq V, |S| \leq n / 2} \frac{w(S, V - S)}{|S|}
= \min_{S \subseteq V, |S| \leq n / 2}~ 1 - \frac{w(S, S)}{|S|}.
\]
We define the combinatorial gap as 
\begin{equation} \label{e:cgap}
\cgap(G) := \min_{S \subseteq V, T \subseteq V, |S| = |T| \leq n / 2}~ 1 - \frac{w(S, T)}{|S|}.
\end{equation}
Note that $\cgap(G)$ is small if there exists a near-bipartite component.
We prove the following combinatorial analog of the spectral analysis of mixing time.
\begin{theorem} \label{t:cgap}
For any graph $G$, $\psi(G) \geq \Omega(\cgap^2(G))$.
\end{theorem}

By the aforementioned result of Morris and Peres~\cite{MP}, one immediate corollary is that the mixing time of non-lazy random walks is upper bounded by $O(\log(n)/\cgap^2(G))$.
An implication is that adding self-loops of weight $\phi(G)$ (instead of $1/2$) is enough to guarantee mixing time $O(\log(n)/\phi^2(G))$ in any graph,
which may have applications in speeding up random sampling algorithms.

Our proof of Theorem~\ref{t:cgap} is based on a new analysis of the approach by Lov\'asz and Simonovits~\cite{LS} using the bar chart in Figure~\ref{f:barchart}.
We believe that the new analysis is more intuitive and provides better insights into what combinatorial properties are needed for fast mixing.

Using Theorem~\ref{t:cgap} and the results in~\cite{KL14},
another corollary is the following lower bound on small-set expansion of graph powers.

\begin{corollary} \label{c:power}
For any graph $G$ and any integer $t \geq 1$,
\[\phi_{\delta/4}(G^t) = \Omega(\min\{ \sqrt{t} \cdot \phi_{\delta}(G), 1 \}).\]
\end{corollary}

The same result is proved in~\cite{KL14} for lazy graphs, and here we prove it for all graphs.
Note that it is not true that $\phi(G^t) = \Omega(\min\{\sqrt{t} \cdot \phi(G), 1\})$ when $G$ is bipartite, but the above corollary shows that it is true for small-set expansion even when $G$ is bipartite.
As shown in~\cite{RS14}, this result can be used to amplify hardness results for the small-set expansion problem.

\subsubsection*{Vertex Expansion}

The robust vertex expansion is defined by Kannan, Lov\'asz and Montenegro~\cite{KLM} as follows:
For $S \subseteq V$, let $N_{1/2}(S) := \min\{ |T|~|~ T \subseteq V-S {\rm~and~} w(S,T) \geq \frac12 w(S,V-S)\}$.
Define 
\[\ver(S) := \min \left\{ \frac{N_{1/2}(S)}{|S|}, 1 \right\} \quad {\rm and} \quad
\ver(G) := \min_{S:|S| \leq n/2} \ver(S)\]
as the robust vertex expansion of a set $S$ and the graph $G$.
This definition slightly differ from the original definition in~\cite{KLM} by bounding the vertex expansion above by one.
We do this because it is the range of interest and the statement of our result would be much cleaner.
Also define
\[\Psi(S) := \phi(S) \cdot \ver(S) \quad {\rm and}
\quad \Psi(G) := \min_{S:|S| \leq n/2} \Psi(S)\]
as the minimum product of the edge expansion and the robust vertex expansion.
It is proved in~\cite{KLL} that 
\[1-\lambda_2 = \Omega(\Psi(G)),\]
and that the spectral partitioning algorithm and the local graph partitioning algorithm using personal pagerank vectors~\cite{ACL} achieve better approximation when the robust vertex expansion is large.
We prove a similar result for random walks and evolving sets.

\begin{theorem} \label{t:vertex}
For lazy graphs $G$, $\psi(G) \geq \Omega(\Psi(G))$.
\end{theorem}

A corollary is an improved analysis of Theorem~\ref{t:approx} when the robust vertex expansion of $G$ is large.

\begin{corollary} \label{c:approx}
For any target set $S^*$ and any $\eps > 0$, there is a subset $S'$ with $|S'| \geq |S^*|/2$, such that if we start the evolving set process with $\{v\}$ for $v \in S'$,
then with constant probability the algorithm returns a set $S$ with $\Psi(S) = O(\phi(S^*)/\eps)$ and $|S| = O(|S^*|^{1+\eps})$, and the running time is $O(|S^*|^{1+2\eps} \phi(S)^{-\frac12} \log^2 n)$.
\end{corollary}

Note that the conclusion $\Psi(S) = O(\phi(S^*)/\eps)$ implies that $\phi(S) = O(\phi(S^*)/\phi^V(S)) = O(\phi(S^*)/\phi^V(G))$.
In particular, this implies that the evolving set algorithm is a constant factor approximation algorithm when $\phi^V(G)$ is a constant, for example when $G$ is a planted random graph.
This shows that the evolving set algorithm matches the improved analysis of the spectral partitioning algorithm in~\cite{KLL}.
We refer the reader to~\cite{KLL} for more discussions and motivations for robust vertex expansion.

\subsubsection*{Limitations}

The subexponential time algorithm for small-set expansion by Arora, Barak and Steurer~\cite{ABS} uses eigenspace enumeration and random walks.
The short code example in~\cite{BGHMRS} shows the limitation of the eigenspace enumeration method.
It is a natural question to ask whether random walks can be used to disprove the small-set expansion hypothesis.
There were very few results showing the limitations of these random walks based algorithms (see~\cite{ZLM} for the only such result that we know of).
One main difference between the truncated random walk algorithm by Spielman and Teng~\cite{ST} and the evolving set algorithm by Andersen and Peres~\cite{AP} and Oveis Gharan and Trevisan~\cite{OT} is that the random walk algorithm is deterministic while the evolving set algorithm involves much randomness.
Oveis Gharan~\cite{O} conjectured in his thesis (Conjecture 12.3.4) that there is a small but nontrivial probability that all the sets explored by the evolving set process is of size $O(|S^*|)$, and argued that this would disprove the small-set expansion hypothesis.
We present an example for which the evolving set algorithm fails with
probability one, refuting Oveis Gharan's conjecture.

\begin{theorem} \label{t:hypercube}
Given any $\eps$, there exists $\delta_\eps$ such that for any $\delta > 0$,
there is a graph $G$ such that $\phi_\delta(G) \leq \eps$, but any subset of
volume $\leq \delta_\eps n$ returned by the evolving set algorithm
in~\cite{AP,OT} has expansion at least $1-\eps$ with probability one.
\end{theorem}

The example is a $k$-ary $\eps$-noisy hypercube, where the dimension cuts are of size $n/k$ with expansion $\eps$.
We show that, however, the evolving set algorithm will only explore the Hamming balls, and the expansion is at least $1-\eps$ for all Hamming balls of size $O(n/k)$.
We note that this example also shows that the random walk algorithm~\cite{ST,KL12} and the pagerank algorithm~\cite{ACL,ZLM} fail to disprove the small-set expansion hypothesis; see Section~\ref{s:limitations}.

We believe that this example exposes the limitations of all known local graph partitioning algorithms, and can be used as a basis to prove further lower bounds (e.g.~to show that the analysis of the $O(\sqrt{\phi(S) \log(|S|)})$-approximation of the evolving set algorithm in Theorem~\ref{t:approx} is tight when $\eps = 1/\log(|S|)$).

\subsection{Relations with Previous Work}

\subsubsection*{Combinatorial Analog of Spectral Gap}

We note that the original analyses of Lov\'asz and Simonovits~\cite{LS} and Andersen and Peres~\cite{AP} heavily rely on the laziness assumption and cannot be used to work with the combinatorial gap.
The bar chart in Figure~\ref{f:barchart} is the new element introduced to analyze the combinatorial gap as well as the robust vertex expansion.

Trevisan~\cite{T} defined the bipartiteness ratio $\beta(G)$ of a graph and proved that $\lambda_n+1$ is related to $\beta(G)$ as if $1-\lambda_2$ is related to $\phi(G)$ stated by Cheeger's inequality.
After formulated and proved Theorem~\ref{t:cgap},
we observe that Trevisan's result combined with the spectral argument can also be used to derive the corollary that the mixing time of non-lazy random walks is bounded by $O(\log(n)/\cgap(G)^2)$.
However, we remark that Theorem~\ref{t:cgap} and Corollary~\ref{c:power} cannot be derived from Trevisan's result and the spectral approach, and also that the formulation of Theorem~\ref{t:cgap} is new.

Bilu and Linial~\cite{BL} defined a combinatorial property called ``jumbleness'', proved that it is a $\log(d)$-approximation to the spectral gap where $d$ is the maximum degree of the graph, and used it to establish a converse to the expander mixing lemma.
The definition of the jumbleness is similar to our definition of the combinatorial gap in that it also concerns about $w(S,T)$ for two subsets of vertices $S,T \subseteq V$, but the precise definition and the theorem obtained are incomparable to what we have in this paper.

\subsubsection*{Mixing Time and Local Graph Partitioning}

The results in Lov\'asz and Kannan~\cite{LK} and Kannan, Lov\'asz and Montenegro~\cite{KLM} show that the mixing time of lazy random walks is $O(\log(n)/\Psi(G))$, among other conditions that imply faster mixing.
However, their results cannot be applied to analyze local graph partitioning algorithms as in Corollary~\ref{c:approx}.

Besides the random walk algorithm~\cite{ST,KL12} and the evolving set algorithm~\cite{AP,OT}, there is also a local graph partitioning algorithm using pagerank vectors~\cite{ACL,ZLM}.
In terms of the approximation guarantee, the output size, and the running time, the pagerank algorithm is subsumed by the evolving set algorithm in~\cite{AP,OT}.

In~\cite{KLL}, it was shown that the pagerank algorithm performs better when the robust vertex expansion is large. 
Similar results were not known for random walks and evolving sets, as the spectral techniques in~\cite{KLL} are not applicable.
These results are proved in this paper by a new analysis of the combinatorial approach of Lov\'asz and Simonovits~\cite{LS}.
Finally, we remark that this paper is a subsequent work of~\cite{KLL}, and both the results and the techniques are different from~\cite{KLL}, especially the combinatorial analog of spectral gap, the counterexample for the evolving set algorithm, and the new analysis of Lov\'asz and Simonovits approach using the barchart.

\section{Faster Convergence}

In this section, we prove the positive results about faster convergence rates of random walks and evolving sets.
Our proofs are based on the combinatorial method of Lov\'asz and Simonovits~\cite{LS}, and we will begin with an introduction of their techniques in Section~\ref{s:curve}, and then we will discuss the proof outline and highlight the new idea in Section~\ref{s:outline}.
Then, we will prove Theorem~\ref{t:cgap} about combinatorial analog of spectral gap in Section~\ref{s:cgap} and then prove Corollary~\ref{c:power} about small-set expansion of graph powers.
Then, we will prove Theorem~\ref{t:vertex} about robust vertex expansion in Section~\ref{s:vertex} and then show its application in local graph partitioning.
Finally, we will discuss how to extend the results to non-regular graphs in Section~\ref{s:general}.

\subsection{Lov\'asz-Simonovits Curve} \label{s:curve}

For any vector $p \in \mathbb R^n$, Lov\'asz and Simonovits~\cite{LS} study the curve $C(p) : [0, n] \to \mathbb R$ that plots the cumulative sum of $p$ defined as
\begin{equation} \label{e:curve}
C(p,x)
= \max_{c \in [0, 1]^n : \sum_i c(i) = x} \sum_{i \in V} c(i) \cdot p(i).
\end{equation}
In words, $C(p,x)$ is just the sum of the first $x$ largest elements in $p$ when $x$ is a positive integer, and the curve $C(p,x)$ is defined for all $x \in [0,n]$ by piecewise linear extension.
It is clear from the definition that $C(p)$ is a concave function.
We are interested in studying the curve $C(A^t p, x)$ where $A$ is a random walk matrix and $p$ is a probability distribution.
Notice that as $A^t p$ converges to the uniform distribution as $t$ becomes larger, $C(A^t p)$ converges to the line $x/n$ and vice versa.
In~\cite{LS}, their method to bound the mixing time is to bound the difference between $C(A^t p)$ and the line $x/n$.
When $A$ is the lazy random walk matrix,
they proved that
\begin{equation} \label{e:converge}
C(A^t p, x) \leq \frac{x}{n} + \sqrt{x} (1 - \frac{\phi(G)^2}{8})^t,
\end{equation}
and this implies that the mixing time of lazy random walks is $O(\log n / \phi(G)^2)$.
The key lemma in their proof is the following inequality:
For any lazy random walk matrix $A$, any $p \in \R^n$ and any integral $x$,
\begin{equation} \label{e:chord}
C(Ap, x) \leq \frac12 \big(C(p, x(1 - \phi(G))) + C(p, x(1 + \phi(G))) \big).
\end{equation}
The bound in (\ref{e:converge}) follows from an inductive argument using (\ref{e:chord}); see~\cite{LS,ST,KL12} and also a slightly more general version in Lemma~\ref{l:lovasz} in Section~\ref{s:vertex}. 
We remark that their proof of (\ref{e:chord}) crucially relies on the assumption that there is a self-loop of weight $1/2$ on each vertex and is a bit magical.

\subsection{Proof Outline} \label{s:outline}

We mainly outline the proof of Theorem~\ref{t:cgap} in this subsection, but we will briefly mention the modifications to prove Theorem~\ref{t:vertex} at the end.
We will prove the following inequality similar to (\ref{e:chord}) using the combinatorial gap (without the laziness assumption that $A_{ii} \geq 1/2$ for all $i \in V$).

\begin{lemma}
\label{l:comb_drop}
For any random walk matrix $A$, any $p \in \R^n$ and any integral $x \leq n / 2$,
\[
C(Ap, x)
\leq \frac12 \big(C(p, x(1 - \cgap(G))) + C(p, x(1 + \cgap(G)))\big).
\]
\end{lemma}

With Lemma~\ref{l:comb_drop} in place of~(\ref{e:chord}),
the same inductive argument that we mentioned before implies the convergence result in~(\ref{e:converge}) with $\phi(G)$ replaced by $\cgap(G)$.
It turns out that the analysis of the Lov\'asz-Simonovits curve can be used to analyze the evolving set process,
and the arguments in Lemma~\ref{l:comb_drop} can be adapted to prove Theorem~\ref{t:cgap}.
To prove Lemma~\ref{l:comb_drop}, 
we consider an arbitrary $S \subseteq V$ and try to bound the total probability in $S$ after one step of random walk $(Ap)(S) := \sum_{i \in S} (Ap)_i$.
To bound $(Ap)(S)$, we look at where the probability in $S$ is coming from.
For each $i \in V$, let 
\[d_S(i) := w(i,S)\] 
be the total weight coming from $i$ to $S$.
Recall that we assume the weighted degree of each vertex is one.
So, we have $d_S(i) \in [0, 1]$ for any $i$ and
\begin{equation} \label{e:sum}
\sum_{i \in V} d_S(i) = |S|.
\end{equation}
The reason of this definition is that $(Ap)(S) = \sum_{i=1}^n d_S(i) \cdot p(i)$.
We sort the vertices so that $d_S(1) \geq d_S(2) \geq \dots \geq d_S(n)$.
Let $T := \{1,2,\ldots,|S|\}$ be the $|S|$ vertices with largest $d_S$ values; 
note that $T$ is in general not equal to $S$.
See Figure~\ref{f:barchart} for an illustration of the proof setup.

\begin{figure} \label{f:barchart}
\begin{center}
\includegraphics[scale=.55]{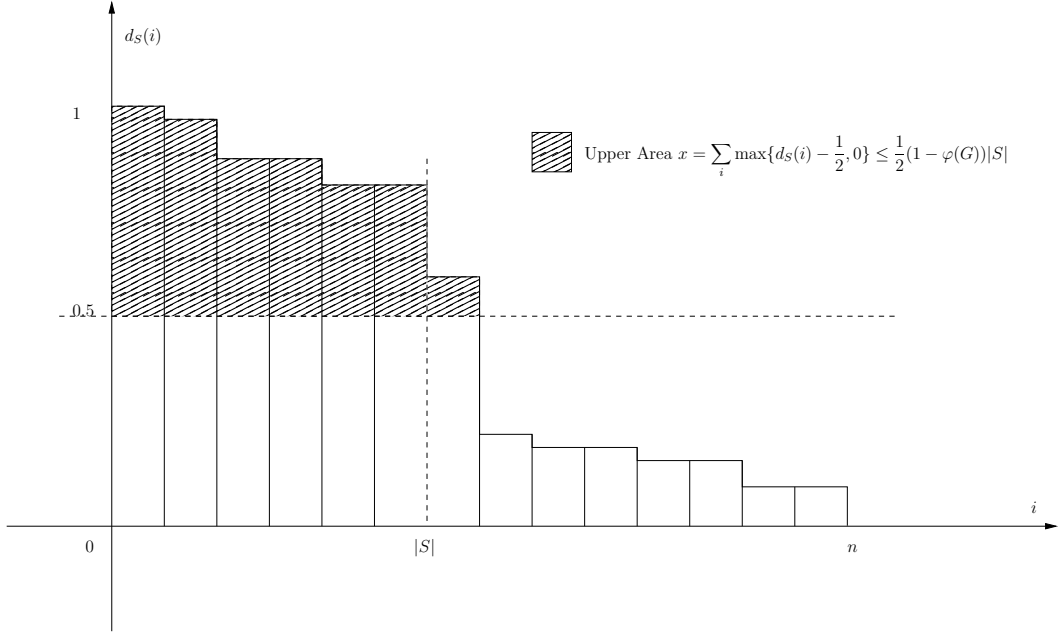} 
\caption{An illustration of the proof ideas.}
\label{fig}
\end{center}
\end{figure}

The obstruction for mixing is when $d_S(1) = d_S(2) = \ldots = d_S(|S|) = 1$
and $p(1), \ldots, p(|S|)$ are the highest probability in $p$,
in which case we would have $C(Ap,|S|) \geq (Ap)(S) = \sum_{i \in S} p(i) = C(p,|S|)$ and thus the curve is not dropping after one step of random walk.
This could happen when $S = T$ in which case $S$ is a disconnected component of $G$ (corresponding to $\lambda_2 = 1$), or when $S$ and $T$ form a bipartite component (corresponding to $\lambda_n = -1$), and in these cases the random walk may not mix (depending on the initial distribution).

The combinatorial gap in~(\ref{e:cgap}) is defined precisely to exclude the obstruction.
It states that any subset $T$ with $|T|=|S|$ can only contribute $w(S,T) \leq (1-\cgap(G))|S|$ to $w(S,V)=|S|$, so as to guarantee that the curve would drop, i.e. $C(Ap,|S|) = \max_{S':|S'|=|S|} Ap(S') < C(p,|S|)$.
To prove Lemma~\ref{l:comb_drop},
we look at the bar chart in Figure~\ref{f:barchart} horizontally and consider the telescoping sum
\[
(Ap)(S) 
= \sum_{i=1}^n d_S(i) \cdot p(i) 
= \sum_{i=1}^n (d_S(i) - d_S(i+1)) \sum_{j=1}^i p(j) 
\leq \sum_{i=1}^n (d_S(i) - d_S(i+1))) \cdot C(p,i).
\]
Suppose we put a threshold $1/2$ and consider the upper area $x = \sum_{i=1}^n \max\{d_S(i) - \frac12, 0\}$ and the lower area $y = \sum_{i=1}^n \min\{d_S(i), \frac12\}$; see Figure~\ref{f:barchart}.
By concavity of the curve $C$,
we will prove in Lemma~\ref{l:drop_by_upper_area} that
\[
(Ap)(S)
\leq \frac12 (C(p, 2 x) + C(p, 2 y))
= \frac12 (C(p, 2 x) + C(p, 2 |S| - 2 x)).
\]
The definition of combinatorial gap in~(\ref{e:cgap}) forces $d_S$ to spread out, and we will prove in Lemma~\ref{l:upper_area_bound} that it implies that
\[x \leq \frac{1}{2} (1-\cgap(G))|S|.\]
Combining these two steps gives Lemma~\ref{l:comb_drop}.
Once we prove Lemma~\ref{l:comb_drop}, the same calculations can be used to prove Theorem~\ref{t:cgap} about the gauge in the evolving set process.

The proof of Theorem~\ref{t:vertex} is also based on the idea of using the bar chart in Figure~\ref{f:barchart}.
First, we translate the definition of robust vertex expansion to an upper bound of the area of the largest vertices.
Then, we will choose a different threshold, and then we will modify the inductive argument accordingly to establish a similar bound as in~(\ref{e:converge}).
The proof will work for the gauge in the evolving set process to prove Theorem~\ref{t:vertex}.

\subsection{Combinatorial Analog of Spectral Gap} \label{s:cgap}

As described in the outline, we will prove the following two lemmas.

\begin{lemma}
\label{l:drop_by_upper_area}
For any subset $S \subseteq V$ with $|S| \leq n / 2$, let $x = \sum_{i \in V} \max\{d_S(i) - \frac12, 0\}$, then
\[
(Ap)(S)
\leq \frac12 (C(p, 2 x) + C(p, 2(|S| - x))),
\]
for any random walk matrix $A$ and any vector $p \in \R^n$.
\end{lemma}

\begin{lemma}
\label{l:upper_area_bound}
For any subset $S \subseteq V$ with $|S| \leq n / 2$, let $x = \sum_{i \in V} \max\{d_S(i) - \frac12, 0\}$, then
\[
x \leq \frac12 (1 - \cgap(G))|S|.
\]
\end{lemma}

First, we assume the lemmas are correct and derive Lemma~\ref{l:comb_drop}.

\begin{proof}[Proof of Lemma~\ref{l:comb_drop}]
Let $S$ be a subset of $V$.
By Lemma~\ref{l:drop_by_upper_area}, Lemma~\ref{l:upper_area_bound} and the concavity of $C$, we have
\[
(Ap)(S)
\leq \frac12 (C(p, 2 x) + C(p, 2(|S| - x)))
\leq \frac12 (C(p, (1 - \cgap(G))|S|) + C(p, (1 + \cgap(G))|S|)),
\]
where the last inequality holds since for any concave function $f$ and any values $a \leq b \leq c \leq d$ with $a + d = b + c$, we have $f(a) + f(d) \leq f(b) + f(c)$.
Since the argument applies to all subset $S \subseteq V$ with $|S| \leq n/2$, we have
\[
C(Ap, |S|)
= \max_{S' \subseteq V, |S'| = |S|} (Ap)(S')
\leq \frac12 \big(C(p, (1 - \cgap(G))|S|) + C(p, (1 + \cgap(G))|S|)\big).
\]
\end{proof}

Next we prove Lemma~\ref{l:drop_by_upper_area}, which follows from the concavity of the curve $C$.

\begin{proof}[Proof of Lemma~\ref{l:drop_by_upper_area}]
Recall that $(Ap)(S) = \sum_i p(i) \cdot d_S(i)$.
For convenience, we define the boundary values to be 
$d_S(0) = 1$ and $d_S(n + 1) = 0$, and also define $p(0)=0$.
Then, there is an index $k$ such that $d_S(k) > 1 / 2 \geq d_S(k + 1)$.
By looking at the bar chart in Figure~\ref{f:barchart} horizontally as in the outline and considering the telescoping sum, we have
\begin{align*}
(Ap)(S)
&= \sum_{i = 0}^n d_S(i) \cdot p(i)
= \sum_{i = 0}^n (d_S(i) - d_S(i + 1)) (\sum_{j = 1}^i p(j))  
\leq \sum_{i = 0}^n (d_S(i) - d_S(i + 1)) \cdot C(p, i)\\
&= \left( \sum_{i = 0}^{k - 1} (d_S(i) - d_S(i + 1)) \cdot C(p, i)  + (d_S(k) - \frac12) \cdot C(p, k) \right)\\
&\hspace{1in} + \left( (\frac12 - d_S(k + 1)) \cdot C(p, k) + \sum_{i = k + 1}^n (d_S(i) - d_S(i + 1)) \cdot C(p, i) \right).
\end{align*}
Recall that $C(p)$ is concave, and for any concave function $f$, we have by Jensen's inequality
\[
\sum_i a_i f(i)
\leq (\sum_i a_i) f\big(\sum_i \frac{a_i}{\sum_i a_i} i\big).
\]
Since 
\[\sum_{i = 0}^{k - 1} (d_S(i) - d_S(i + 1)) + (d_S(k) - \frac12 ) = d_S(0) - \frac12 = \frac12\] 
and 
\[(\frac12 - d_S(k+1)) + \sum_{i = k+1}^{n} (d_S(i) - d_S(i + 1)) = \frac12 - d_S(n+1) = \frac12,\] 
by applying Jensen's inequality with $f=C(p)$, we have
\begin{align*}
(Ap)(S)
&\leq \frac12 C \left( p, 2 \Big(\sum_{i = 0}^{k - 1} (d_S(i) - d_S(i + 1)) \cdot i + (d_S(k) - \frac12) \cdot k \Big) \right)\\
&\hspace{1in} + \frac12 C \left( p, 2 \Big( (\frac12 - d_S(k + 1)) \cdot k + \sum_{i = k + 1}^n (d_S(i) - d_S(i + 1)) \cdot i \Big) \right).
\end{align*}
Finally, note that the first sum
\[
\sum_{i = 0}^{k - 1} (d_S(i) - d_S(i + 1)) \cdot i + (d_S(k) - \frac12) \cdot k
= \sum_{i = 1}^k (d_S(i) - \frac12)
= \sum_{i = 1}^n \max(d_S(i) - \frac12, 0)
= x.
\]
Since $\sum_{i = 0}^n (d_S(i) - d_S(i + 1)) \cdot i = \sum_{i = 1}^n d_S(i) = |S|$ by (\ref{e:sum}), the second sum is $|S|-x$, and so we have
\[
(Ap)(S)
\leq \frac12 C(p, 2 x) + C(p, 2(|S| - x)).
\]
\end{proof}

Now we prove Lemma~\ref{l:upper_area_bound}, which uses the definition of the combinatorial gap in (\ref{e:cgap}).

\begin{proof}[Proof of Lemma~\ref{l:upper_area_bound}]
Recall that we sort the vertices such that $1 \geq d_S(1) \geq \ldots \geq d_S(n) \geq 0$.
Let $k$ be the index such that $d_S(k) > 1 / 2 \geq d_S(k + 1)$.
Let $T = \{1,\ldots,|S|\}$ be the subset of the first $|S|$ vertices.
We consider two cases.
The first case is when $k \leq |S|$, in which
\[
x
= \sum_{i=1}^n \max\{d_S(i) - \frac12,0\}
= \sum_{i = 1}^k (d_S(i) - \frac12)
\leq \sum_{i = 1}^k \frac{d_S(i)}2
\leq \sum_{i = 1}^{|S|} \frac{d_S(i)}2
\leq \frac12 (1 - \cgap(G)) |S|,
\]
where the first inequality holds as $d_S(i) \in [0,1]$,
and the last inequality holds by using (\ref{e:cgap}) to obtain 
$\cgap(G) \leq 1 - w(S,T)/|S| = 1 - \sum_{i=1}^{|S|} d_S(i) /|S|$.

The second case is when $k > |S|$.
Note that $d_S(|S|) \leq 1 - \cgap(G)$, as otherwise $T$ would violate~(\ref{e:cgap}).
Then, for any $i \geq |S|$, we have $d_S(i) \leq d_S(|S|) \leq 1 - \cgap(G)$.
Therefore, for any $i \geq |S|$,
\[
d_S(i) - \frac12
\leq d_S(i) \big( \frac{\frac12 - \cgap(G)}{1 - \cgap(G)} \big)
\leq \frac{d_S(i) \cdot (1 - \cgap(G))}2,
\]
and we have
\begin{align*}
x
&= \sum_{i = 1}^k (d_S(i) - \frac12)
\leq \sum_{i = 1}^{|S|} (d_S(i) - \frac12) + \sum_{i = |S| + 1}^k (d_S(i) - \frac12)
\leq \sum_{i = 1}^{|S|} (d_S(i) - \frac12) + (\frac{1-\cgap(G)}{2}) \sum_{i = |S| + 1}^k d_S(i)\\
& = \big( w(S,T) - \frac{|S|}2 \big) + (\frac{1-\cgap(G)}{2}) (|S| - w(S,T))
= (\frac{1 + \cgap(G)}2) w(S,T) - \frac{|S|}2 + (\frac{1 - \cgap(G)}2) |S|\\
&\leq (\frac{1 + \cgap(G)}2) (1-\cgap(G)) |S| - \frac{|S|}2 + (\frac{1 - \cgap(G)}2)|S|
\leq (\frac{1 - \cgap(G)}2) |S|,
\end{align*}
where the second last inequality is by~(\ref{e:cgap}).
Hence, in any case, we have $x \leq \frac12 (1 - \cgap(G)) |S|$.
\end{proof}

\subsubsection*{Evolving Sets}

Finally, we prove Theorem~\ref{t:cgap} about the gauge of the evolving set process, which involves very similar calculations as in the proof of Lemma~\ref{l:comb_drop}.
Recall that the gauge is defined as $\psi(S) = 1 - \expe[\sqrt{|\tilde S| / |S|}]$.
The following claim follows from Morris and Peres~\cite{MP}.

\begin{claim}[\cite{MP}, Equation 27] \label{c:MP}
For any $t \in [0, 1]$, we have
\[
t \cdot \expe[|\tilde S|~ |~U \leq t]
= \sum_{i \in V} \min\{t, d_S(i)\}.
\]
\end{claim}
\begin{proof}
Note that $\prob[i \in \tilde S~|~U \leq t] = \min\{1, d_S(i) / t\}$.
It follows that 
\[\expe[|\tilde S|~|~U \leq t] = \sum_{i \in V} \prob[i \in \tilde S~|~U \leq t] = \sum_{i \in V} \min\{1, d_S(i) / t\}.\]
\end{proof}

When $t=1/2$, Claim~\ref{c:MP} implies that 
\[
\frac12 \expe[|\tilde S|~|~U \leq \frac12] = \sum_{i \in V} \min\{\frac12, d_S(i)\},
\]
which is the lower area under the threshold $1/2$ (see Figure~\ref{f:barchart}).
And 
\[
\frac12 \expe[|\tilde S|~|~U \geq \frac12] 
= \expe[\tilde S] - \frac12 \expe[|\tilde S|~|~U \leq \frac12]
= \sum_{i \in V} d_S(i)- \sum_{i \in V} \min\{\frac12, d_S(i)\}
= \sum_{i \in V} \max\{d_S(i)-\frac12,0\},
\]
which is the upper area above the threshold $1/2$, and this is at most $\frac12 (1 - \cgap(G)) |S|$ by Lemma~\ref{l:upper_area_bound}.
We are ready to prove Theorem~\ref{t:cgap}.

\begin{proof}[Proof of Theorem~\ref{t:cgap}]
By definition,
\[
\psi(S)
= 1 - \expe[\sqrt{|\tilde S| / |S|}]
= 1 - \frac12 \expe[\sqrt{|\tilde S|/|S|}~|~U \leq \frac12] - \frac12 \expe[\sqrt{|\tilde S|/|S|}~|~U \geq \frac12].
\]
Let $x = \sum_i \max\{d_S(i) - \frac12, 0\}$ be the upper area.
We know that $x \leq \frac12 (1 - \cgap(G)) |S|$ by Lemma~\ref{l:upper_area_bound}.
By concavity,
\[
\frac12 \expe[\sqrt{|\tilde S|/|S|}~|~U \geq \frac12]
\leq \frac12 \sqrt{ \frac{1}{|S|} \expe[|\tilde S|~|~U \geq \frac12] }
= \frac12 \sqrt{\frac{2x}{|S|}}.
\]
Similarly,
\[
\frac12 \expe[\sqrt{|\tilde S|/|S|}~|~U \leq \frac12]
\leq \frac12 \sqrt{\frac{1}{|S|} \expe[|\tilde S|~|~U \leq \frac12]}
= \frac12 \sqrt{\frac{2(|S|-x)}{|S|}}.
\]
Plugging these two inequalities back into the first one, we have
\[
\psi(S)
\geq 1 - \frac12 (\sqrt{2 - \frac {2 x}{|S|}} + \sqrt{\frac{2 x}{|S|}})
\geq 1 - \frac12 (\sqrt{1 + \cgap(G)} + \sqrt{1 - \cgap(G)})
\geq \cgap(G)^2 / 8,
\]
where the second inequality is because $\sqrt{2 - 2x/|S|} + \sqrt{2x/|S|} \leq \sqrt{2 - 2y/|S|} + \sqrt{2y/|S|}$ for $x \leq y \leq |S|/2$ and we put in $y = \frac12 (1 - \cgap(G))|S|$,
and the last inequality is by Taylor expansion of the function $\sqrt{1+\cgap}$.
\end{proof}

\subsubsection*{Small-Set Expansion of Graph Powers} \label{s:power}

In this subsection, we prove Corollary~\ref{c:power} about small-set expansion of graph powers.
Recall that the $\delta$-small-set expansion is defined as
\[
\phi_\delta(G)
= \min_{S \subseteq V, |S| \leq \delta n} 1 - \frac{w(S, S)}{|S|}.
\]
We define the $\delta$-small-set combinatorial gap as
\[
\cgap_\delta(G)
= \min_{S \subseteq V, T \subseteq V, |S| = |T| \leq \delta n} 1 - \frac{w(S, T)}{|S|}.
\]
The following is a simple relation between these two quantities.
\begin{lemma} \label{l:relation}
For any $\delta \leq 1 / 2$, 
we have $\cgap_{\delta / 2}(G) \geq \phi_\delta(G) / 2$.
\end{lemma}

\begin{proof}
Suppose $S, T$ are two subsets of size at most $\delta n/2$ that achieve $w(S, T) = (1 - \cgap_{\delta / 2}(G))|S|$ and $|S|=|T|$.
We argue that $S \cup T$ has small expansion, since
\begin{align*}
w(S \cup T, S \cup T)
&\geq 2w(S, T) - w(S \cap T, S \cap T)
\geq 2 (1 - \cgap_{\delta / 2}(G))|S| - |S \cap T|\\
&= 2|S| - |S \cap T| - 2 |S| \cgap_{\delta / 2}(G)
= |S \cup T| - 2|S| \cgap_{\delta/2}(G)
\geq |S \cup T| (1 - 2 \cgap_{\delta / 2}(G)).
\end{align*}
Since $|S \cup T| \leq 2 |S| \leq \delta n$, this implies that $\phi_\delta(G) \leq 1 - w(S \cup T, S \cup T)/|S \cup T| \leq 2\cgap_{\delta/2}(G)$.
\end{proof}

In~\cite{KL14}, using the inequality~(\ref{e:chord}), it is proved that for any lazy graph,
\[
\phi_{\delta / 2}(G^t)
\geq \Omega(\min(\sqrt t \cdot \phi_\delta(G), 1)).
\]
The assumption of laziness is only used for the inequality~(\ref{e:chord}) to hold.
Now, with Lemma~\ref{l:comb_drop}, we can replace $\phi_\delta(G)$ by $\cgap_\delta(G)$ and prove that for any graph $G$,
\[
\phi_{\delta / 2}(G^t)
\geq \Omega(\min(\sqrt t \cdot \varphi_\delta(G), 1)).
\]
Combining with Lemma~\ref{l:relation}, it follows that for any graph $G$,
\[
\phi_{\delta / 4}(G^t)
\geq \Omega(\min(\sqrt t \cdot \phi_\delta(G), 1)).
\]

\subsection{Vertex Expansion} \label{s:vertex}

We will prove Theorem~\ref{t:vertex} in this subsection.
As in Section~\ref{s:cgap}, we will first prove faster convergence for random walks and then for evolving sets.
Intuitively, larger vertex expansion will lead to faster mixing, as the probability coming into a set $S$ is from many different vertices, and so it cannot be the case that all probability in $S$ come from a small number of vertices with high probability.
Our proof idea is also to look at the bar chart in Figure~\ref{f:barchart}, and translate the definition of robust vertex expansion into an upper bound of the area of the largest vertices.

For $S \subseteq V$, we consider the vector $d_S \in \R^n$ with the $i$-th entry being $d_S(i)$, and as before we assume that $d_S(1) \geq d_S(2) \geq \ldots \geq d_S(n)$.
We use the notation $C(d_S,x)$ in~(\ref{e:curve}) on $d_S$ so that we can talk about the largest $x$ values in the vector $d_S$ (note that $x$ could be non-integral).

In Theorem~\ref{t:vertex}, unlike in Theorem~\ref{t:cgap}, we need the additional assumption that the random walk matrix is lazy, such that $w(i,i) \geq 1/2$ for any $i \in V$.
The main reason of this assumption is to have $d_S(i) \geq 1/2$ for $i \in S$ and $d_S(j) \leq 1/2$ for $j \notin S$, so that we can assume that 
\[S = \{1, 2, \ldots, |S|\},\] 
i.e. the vertices in $S$ are the vertices with the largest values in $d_S$.
Recall that we defined 
\[N_{1/2}(S) := \min\{|T|~|~T \subseteq V-S {\rm~and~} w(S,T) \geq \frac12 w(S,V-S)\}.\]
Since $S = \{1,\ldots,|S|\}$,
we have for any integral $x$,
\[\max_{T:T \subseteq V-S, |T|=x} w(S,T)
= \max_{T:T \subseteq V-S, |T|=x} \sum_{i \in T} d_S(i)
= \sum_{i=|S|+1}^{|S|+x} d_S(i)
= C(d_S, |S|+x) - C(d_S, |S|).\] 
In words, the set of vertices $T \subseteq V-S$ that maximize $w(S,T)$ are the vertices $\{|S|+1, \ldots, |S|+|T|\}$ in the ordering defined by $d_S$.
So, we can rewrite the definition of $N_{1/2}(S)$ as
\[
N_{1 / 2}(S)
:= \min \{ x~|~C(d_S, |S| + x) - C(d_S, |S|) \geq \frac12 w(S, V-S) \}.
\]
Note that we allow $x$ to be non-integral.
This differs by at most one compared with the original definition in~\cite{KLM}, and will make our proofs much cleaner.
The robust vertex expansion is defined as $\phi^V(S) := \min \{ N_{1 / 2}(S) / |S|, 1 \}$ and $\phi^V(G) := \min_{S : |S| \leq |V| / 2} \phi^V(S)$ as before.
Similarly, $\Psi(S)$ is defined as before using the new definition of $N_{1/2}(S)$.
The following lemma translates the definition of $\phi^V(S)$ to a bound on the cumulative sum of the largest vertices in the bar chart.

\begin{lemma} \label{l:vertex-profile}
For any $S \subseteq V$ with $|S| \leq |V| / 2$, we have
\[
C(d_S, (1 + \phi^V(S)) |S|) \le (1 - \frac{\phi(S)}2) |S|.
\]
\end{lemma}

\begin{proof}
Since $C(d_S, |S| + x)$ is continuous with respect to $x$, 
the minimum in the definition of $N_{1 / 2}(S)$ is attained 
when $C(d_S, |S| + x) - C(d_S, |S|) = \frac12 w(S, V-S)$.
As $x = N_{1/2}(S) \le \phi^V(S) |S|$, we have
\[C(d_S, (1 + \phi^V(S)) |S|) - C(d_S, |S|) \le \frac12 w(S, V-S).\]
Finally, since $S=\{1,\ldots,|S|\}$, we have $C(d_S, |S|) = w(S,S) = (1 - \phi(S)) |S|$ and $w(S, V-S) = \phi(S) \cdot |S|$, we have
\[
C(d_S, (1 + \phi^V(S)) |S|) \le  (1 - \frac{\phi(S)}2) |S|.
\]
\end{proof}

Using Lemma~\ref{l:vertex-profile},
we will prove a bound similar to that of Lemma~\ref{l:comb_drop}.
The two steps (Lemma~\ref{l:drop_by_upper_area} and Lemma~\ref{l:upper_area_bound}) of proving Lemma~\ref{l:comb_drop} are integrated and steamlined in the proof of the following lemma. 

\begin{lemma}
\label{l:comb_drop_vertex}
Assume $C(d_S, a |S|) \leq b |S|$ for any $|S| \leq n/2$ for some $a > 1$ and $b < 1$, then for any $p \in \mathbb R^n$, we have for any $S$ with $|S| \leq n/2$,
\[
C(Ap, |S|)
\leq \big( \frac{a - b}{a - b^2} \big) \cdot C(p, b |S|) + \big( \frac{b - b^2}{a - b^2} \big) \cdot C(p, \frac{a |S|}b).
\]
\end{lemma}
\begin{proof}
Using the same concavity argument as in Lemma~\ref{l:drop_by_upper_area}, 
for any threshold $t \in (0, 1)$ 
and $k$ such that $d_S(k) > t \geq d_S(k + 1)$, we have for any $S$ with $|S| \leq n/2$,
\begin{align*}
(Ap)(S)
&\leq (1 - t) \cdot C \left( p, \frac{1}{1 - t} \Big(\sum_{i = 0}^{k - 1} (d_S(i) - d_S(i + 1)) \cdot i + (d_S(k) - t) \cdot k \Big) \right)\\
&\hspace{1in} + t \cdot C \left( p, \frac1t \Big( (t - d_S(k + 1)) \cdot k + \sum_{i = k + 1}^n (d_S(i) - d_S(i + 1)) \cdot i \Big) \right).
\end{align*}
Let $x = \sum_{i = 0}^{k - 1} (d_S(i) - d_S(i + 1)) \cdot i + (d_S(k) - t) \cdot k$ be the upper area above the threshold $t$.
Then, it follows that
\[
C(p,|S|) = \max_{S} (Ap)(S)
\leq (1 - t) \cdot C(p, \frac{x}{1-t}) + t \cdot C(p, \frac{|S| - x}{t}).
\]

It remains to prove an analog of Lemma~\ref{l:upper_area_bound} to bound the upper area $x$.
We again consider two cases.
The first case is when $k \leq a |S|$, in which
\[
x
= \sum_{i = 1}^k (d_S(i) - t)
\leq \sum_{i = 1}^k (1-t) \cdot d_S(i)
= (1-t) \cdot C(d_S, a |S|)
\leq (1 - t) \cdot b|S|,
\]
where the last inequality uses the assumption that $C(d_S, a|S|) \leq b|S|$.
The second case is when $k > a |S|$.
Note that for any $i > a |S|$, $d_S(i) \leq d_S(a |S|) \leq b / a$,
as otherwise the assumption $C(d_S,a|S|) \leq b|S|$ would be violated.
Hence, for any $i > a|S|$,
\[d_S(i) - t \leq \big( \frac{b/a - t}{b/a} \big) \cdot d_S(i) = (1 - \frac{at}{b}) \cdot d_S(i),\] 
and we get
\begin{align*}
x &= \sum_{i = 1}^k (d_S(i) - t) 
\leq C(d_S, a |S|) - a |S| t + \sum_{i = a |S| + 1}^k (1 - \frac{a t}b) \cdot d_S(i) \\
&= C(d_S, a |S|) - a |S| t + \big(|S| - C(d_S, a |S|)\big) (1 - \frac{a t}b)
\leq |S| (1 - \frac{a t}b),
\end{align*}
where the last inequality uses the assumption that $C(d_S, a|S|) \leq b|S|$ and some simple calculations.
We balance the two upper bounds $b |S| (1 - t)$ and $|S| (1 - a t / b)$ by choosing $t = (b - b^2) / (a - b^2)$,
so that in both cases we have 
\begin{equation} \label{e:area}
x \leq \frac{b (a - b)}{a - b^2} |S|.
\end{equation}
Putting the choice of $t$ and the bound on $x$ back (and using concavity), we obtain the conclusion of the lemma.
\end{proof}

Lemma~\ref{l:comb_drop_vertex} is a generalization of (\ref{e:chord}),
and we can use it to derive a generalization of (\ref{e:converge}).
Note that any probability distribution $p$ satisfies the condition $C(p,x) \leq \frac{x}{n} + \sqrt{\min\{x,n-x\}}$ in the following lemma. 

\begin{lemma}
\label{l:lovasz}
Assuming $C(d_S, a |S|) \leq b |S|$ for any $|S| \leq n / 2$ for some $a > 1$ and $b < 1$, 
then for any $p \in \R^n$ satisfying $C(p, x) \leq \frac{x}{n} + c\sqrt{\min\{x,n-x\}}$ for all $x \in [0,n]$ for some $c$,
we have
\[
C(A^t p, x)
\leq \frac{x}{n} + c \sqrt{\min\{x,n-x\}} \cdot (1 - \frac{(\sqrt a - \sqrt b) (1 - \sqrt b)}{\sqrt a + b})^t.
\]
\end{lemma}

\begin{proof}
We only consider the case that $x \leq n / 2$ in the following; the case $x > n/2$ can be handled in the same manner.
When $x \leq n/2$, we have $C(p,x) \leq \frac{x}{n} + c\sqrt{x}$.
By Lemma~\ref{l:comb_drop_vertex},
\begin{align*}
C(Ap, x)
&\leq \big( \frac{a - b}{a - b^2} \big) \cdot C(p, b x) + \big( \frac{b - b^2}{a - b^2} \big) \cdot C(p, \frac{a x}b) \\
&\leq \big( \frac{a - b}{a - b^2} \big) \cdot (\frac{x}{n} + c\sqrt{bx}) + \big( \frac{b - b^2}{a - b^2} \big) \cdot (\frac{x}{n} + c\sqrt{\frac{ax}{b}})\\  
&= \frac{x}{n} + c \sqrt{x} \Big( \big( \frac{a - b}{a - b^2} \big) \cdot \sqrt{b} + \big( \frac{b - b^2}{a - b^2} \big) \cdot \sqrt{\frac{a}{b}} \Big)\\  
&= \frac xn + c \sqrt x \cdot \Big( \frac1{a - b^2} \big( (a - b) + (1 - b) \sqrt a \big) \sqrt b \Big) \\
&= \frac xn + c \sqrt x \cdot \big( \frac1{a - b^2} (\sqrt a - b) (1 + \sqrt a) \sqrt b \big) \\
&= \frac xn + c \sqrt x \cdot \big( \frac{(1 + \sqrt a) \sqrt b}{\sqrt a + b} \big).
\end{align*}

Now note that,
\[
1 - \frac{(1 + \sqrt a) \sqrt b}{\sqrt a + b}
= \frac{\sqrt a + b - \sqrt b - \sqrt{ab}}{\sqrt a + b}
= \frac{(\sqrt a - \sqrt b)(1 - \sqrt b)}{\sqrt a + b}.
\]
The above argument shows that $C(Ap,x) \leq \frac{x}{n} + c \sqrt{x} (1-(\sqrt a-\sqrt b)(1-\sqrt b) / (\sqrt a + b))$.
Apply the same argument inductively (with different $c$ in each iteration) gives the lemma.
\end{proof}

In particular, for any probability distribution $p$, we have the following generalization of~(\ref{e:converge}):
\[
C(A^tp, x) 
\leq \frac{x}{n} + \sqrt{x} \cdot (1 - \frac{(\sqrt a - \sqrt b) (1 - \sqrt b)}{\sqrt a + b})^t.
\]

\subsubsection*{Evolving Sets}

We prove Theorem~\ref{t:vertex} about the gauge in the evolving set process.

\begin{proof}[Proof of Theorem~\ref{t:vertex}]
By Claim~\ref{c:MP}, for any $t \in [0,1]$, the lower area below threshold $t$ is
\[
t \cdot \expe[|\tilde S|~|~U \leq t] 
= \sum_{i \in V} \min\{t, d_S(i)\} 
= |S| - \sum_{i \in V} \max\{d_S(i) - t, 0\},
\]
and the upper area above threshold $t$ is
\[
(1-t) \cdot \expe[|\tilde S|~|~U \geq t] 
= \expe[\tilde S] - t \cdot \expe[|\tilde S|~|~U \leq t]
= \sum_{i \in V} \max\{d_S(i)-t,0\}.
\]
By the same argument in Lemma~\ref{l:comb_drop_vertex}, 
using the assumption $C(d_S, a |S|) \leq b |S|$ 
and setting $t = (b - b^2) / (a - b^2)$, 
the upper area above threshold $t$ is 
\[
x := \sum_{i \in V} \max\{d_S(i) - t, 0\} \leq \frac{b (a - b)}{a - b^2} |S|
\] 
as stated in (\ref{e:area}).
Therefore,
\begin{align*}
\psi(S)
&= 1 - \expe[\sqrt{|\tilde S|/|S|}] \\
&= 1 - t \cdot \expe[\sqrt{|\tilde S|/|S|}~|~U \leq t] - (1 - t) \cdot \expe[\sqrt{|\tilde S|/|S|}~|~U \geq t] \\
&\geq 1 - t \sqrt{\expe[|\tilde S|/|S|~|~U \leq t]} - (1 - t) \sqrt{\expe[|\tilde S|/|S|~|~U \geq t]} \\
&= 1 - t \sqrt{\frac{1}{|S|} \frac{|S| - x}{t} } - (1 - t) \sqrt{\frac{1}{|S|} \frac{x}{1-t} } \\
&\geq 1 - t \sqrt{\frac{1}{t|S|} (|S| - \frac{b (a - b)}{a - b^2}|S| )} - (1 - t) \sqrt{\frac{1}{(1-t)|S|} \frac{b (a - b)}{a - b^2} |S| },
\end{align*}
where the first inequality is by the concavity of the square root function,
and the second inequality is by the following fact: Suppose $f$ is a concave function and $c_1 \geq c_2 \geq c_3 \geq c_4$ satisfy $t c_1 + (1 - t) c_4 = t c_2 + (1 - t) c_3$ for some $t \in [0, 1]$, then $t f(c_1) + (1 - t) f(c_4) \leq t f(c_2) + (1 - t) f(c_3)$.
Note that
\[
\frac{1}{t|S|} (|S| - \frac{b (a - b)}{a - b^2} |S|)
= \frac{a - b^2}{b - b^2} (1 - \frac{b (a - b)}{a - b^2})
= \frac{a - b^2 - b (a - b)}{b - b^2}
= \frac ab,
\]
and
\[
\frac{1}{(1-t)|S|} \frac{b (a - b)}{a - b^2} |S|
= \frac{a - b^2}{a - b} \frac{b (a - b)}{a - b^2}
= b.
\]
Hence,
\[
\psi(S)
\geq 1 - t \sqrt{\frac ab} - (1 - t) \sqrt b
\geq \frac{(\sqrt a - \sqrt b) (1 - \sqrt b)}{\sqrt a + b},
\]
where the last inequality follows from the calculations in Lemma~\ref{l:lovasz}
(starting from the third line in the first block of calculations).
We put $a = 1 + \phi^V(S)$ and $b = 1 - \frac{\phi(S)}{2}$.
Note that since $\phi^V(S) \le 1$ by definition, we have $\sqrt a \ge 1 + \phi^V(S) / 3$.
This is the only place we need to assume $\phi^V(S) \le 1$.
On the other hand $\sqrt b \le 1 - \phi(S) / 2$.
Therefore we have
\[
\psi(S) 
\geq \frac{(\frac13 \phi^V(S) + \frac12 \phi(S)) \frac12 \phi(S)}3
\geq \frac{\phi^V(S) \phi(S)}{18}
= \frac{\Psi(S)}{18}.
\]
\end{proof}

\subsubsection*{Local Graph Partitioning}

We obtain Corollary~\ref{c:approx} about the performance of the evolving set algorithm in~\cite{AP,OT}.
In Lemma~5.2 of~\cite{OT}, Oveis Gharan and Trevisan actually showed that the (volume biased) evolving set process will return a set $S$ with $\psi(S) = O(\phi(S^*)/\eps)$, and they used the fact that $\psi(S) = \Omega(\phi(S)^2)$ to get Theorem~\ref{t:approx}.
Now, with Theorem~\ref{t:vertex}, we can replace $\phi(S)^2$ by $\Psi(S)$ and obtain Corollary~\ref{c:approx}.

\subsection{General Graphs} \label{s:general}

Our results generalize to non-regular undirected graphs, with appropriate changes in various definitions.

{\bf Expansion:}
For a general undirected graph $G$, we use $\vol(S) := \sum_{i \in S} \deg(i)$ to denote the volume of a subset $S$.
It is the non-regular analog of the size $|S|$.
The conductance of a set $S \subseteq V$ and the conductance of the graph are defined as
\[
\phi(S) := \frac{w(S, V - S)}{\vol(S)}
\text{\quad and \quad}
\phi(G) := \min_{S : \vol(S) \le \vol(V) / 2} \phi(S).
\]
Ideally, the analog of the combinatorial gap would be 
\[\varphi(G) := \min_{S, T : \vol(S) = \vol(T) \le \vol(V) / 2} 1 - \frac{w(S, T)}{\vol(S)}.\]
However, this definition may not say much since it can happen that any two different subsets have different volume.
In order to handle this situation, we revise the definition and allow $S$ and $T$ to be fractional.
Let $\vec d$ be the degree vector of $G$.
We define the combintarial gap as
\[
\varphi(G)
:= \min_{\chi_S \in [0, 1]^V, \chi_T \in [0, 1]^V : \inner{\chi_S,\vec d} = \inner{\chi_T, \vec d} \le \vol(V) / 2} 1 - \frac{\inner{\chi_S, A \chi_T}}{\inner{\chi_S, \vec d}}.
\]

{\bf Lov\'asz Simonovits curve:}
In general graphs, the Lov\'asz Simonovits curve $C(p) : \vol(V) \to \mathbb R$ is defined as
\[
C(p, x) := \max_{c \in [0, 1]^n : \inner{ c, \vec d} = x} \inner{c, p}.
\]
Suppose the vertices are sorted so that $p(1) \ge p(2) \ge \dots \ge p(n)$.
The extreme points of the curve $C(p)$ is $\sum_{j = 1}^i \deg(j)$ for $j = 0, \dots n$.

{\bf Bar chart:}
We sort the vertices so that $d_S(i) / \deg(i)$ is decreasing.
We should view the bar chart so that each bar has width $\deg(i)$ and height $d_S(i) / \deg(i)$.
So the total width is $\vol(V)$.
Same as before, we put a threshold $1 / 2$ (or choosing another threshold $t\in [0, 1]$ in the proof for vertex expansion) and consider the upper area $x$ and lower area $y$, and show that
\[
(Ap)(S) \le \frac12 (C(p, 2 x) + C(p, 2 y)).
\]
With this figure in mind, the proofs for the extended results are essentially the same as the original proofs.

{\bf Vertex expansion:}
The robust vertex expansion is defined as follows.
Let $q(i) = d_S(i) / \deg(i)$, and
\[
N_{1 / 2}(S) := \min \{ x ~|~ C(q, \vol(S) + x) - C(q, \vol(S)) \ge \frac{w(S, \bar S)}2 \}.
\]
Then $\phi^V(S) := \min \{ N_{1 / 2}(S) / \vol(S), 1 \}$, and $\phi^V(G) := \min_{S : \vol(S) \le \vol(V) / 2} \phi^V(S)$.

\subsubsection*{Restating the Theorems for General Graphs}

In the following, we restate our results on general graphs without proofs.
Lemma~\ref{l:comb_drop} becomes
\begin{lemma}
For any extreme point $x \le \vol(V) / 2$,
\[
C(Ap, x)
\le \frac12 (C(p, x (1 - \varphi(G))) + C(p, x (1 + \varphi(G)))).
\]
\end{lemma}

Lemma~\ref{l:drop_by_upper_area} becomes
\begin{lemma}
For any subset $S \subseteq V$ with $\vol(S) \le \vol(V) / 2$, let
\[
x
= \sum_{i \in V} \deg(i) \cdot \max \{ \frac{d_S(i)}{\deg(i)} - \frac12, 0 \},
\]
then
\[
(Ap)(S) \le \frac12 (C(p, 2 x), C(p, 2 (\vol(V) - x))),
\]
for any random walk matrix $A$ and any vector $p \in \mathbb R^n$.
\end{lemma}

Lemma~\ref{l:upper_area_bound} becomes
\begin{lemma}
For any subset $S \subseteq V$ with $\vol(S) \le \vol(V) / 2$, let $x = \sum_{i \in V} \deg(i) \max \{ d_S(i) / \deg(i) - 1 / 2, 0 \}$, then
\[
x \le \frac12 (1 - \varphi(G)) \vol(S).
\]
\end{lemma}

Lemma~\ref{l:comb_drop_vertex} becomes
\begin{lemma}
Let $q(i) := d_S(i) / \deg(i)$.
Assume $C(q, a \vol(S)) \le b \vol(S)$ for any $S$ with $\vol(S) \le \vol(V) / 2$ for some $a > 1$ and $b < 1$, then for any $p \in \mathbb R^n$, we have
\[
	C(p, \vol(S)) \le (\frac{a - b}{a - b^2}) \cdot C(p, b \vol(S)) + (\frac{b - b^2}{a - b^2}) \cdot C(p, \frac{a \vol(S)}b).
\]
\end{lemma}

With these definitions and lemmas in place, Theorem~\ref{t:cgap} and Theorem~\ref{t:vertex} hold as stated in the introduction.

\section{Limitations} \label{s:limitations}

We prove Theorem~\ref{t:hypercube} that provides a hard small-set expansion instance for the evolving set process studied in~\cite{AP,OT}.
As mentioned in the introduction, 
it will be a noisy hypercube $H$ over alphabet size $k = 1/\delta$.
Formally, $H$ is a graph on $k^d$ vertices, representing all strings of length
$d$ over alphabet $\domain k$.
For two vertices $x,y$, 
the edge weight $w(x,y)$ is set to be the probability to go from vertex $x$ to vertex $y$ in one step of a random walk, 
where each symbol of $x$ is independently rerandomized with
probability $\eps$:
For each $i\in [d]$, with probability $1-\eps$, set $y_i = x_i$, otherwise
$y_i$ is sampled uniformly at random from $\domain k$.
Note that $H$ is $1$-regular.
It is easy to see that $H$ has a small sparse cut.

\begin{claim} \label{c:dimension}
There is a set $S$ with expansion at most $\eps$ and $|S|=\delta n$ where $\delta=1/k$.
\end{claim}
\begin{proof}
Indeed, the coordinate cut $S = \set{ x\in \domain k^d \mid x_1=0 }$ has size $\delta n$ and expansion at most $\eps$.
\end{proof}

We will show that all the sets explored by the evolving set process have expansion close to one.
First, we argue that the evolving set process will only explore the Hamming balls of the noisy hypercube in Lemma~\ref{l:level}.
Then, we will show that the expansion of all Hamming balls of size $O(\delta n)$ is close to one in Lemma~\ref{l:hamming}.

The evolving set process starts from a singleton set on $H$.
By symmetry, we may assume this set is $\set{0^d}$.
We now show that the evolving set process only explores sets that are Hamming balls $B(r)$ (around $0^d$),
where $B(r)$ denotes all strings of Hamming weight at most $r$:
\[ B(r) := \set{x\in \domain k^d \mid \abs{x} \leq r} {\rm~where~}
\abs{x} := \card{\set{i\in \domain d\mid x_i \neq 0}}. \]

Indeed, the initial set $\set{0^d}$ is the Hamming ball $B(0)$.
The following lemma implies that, if the current set is a Hamming
ball $B(r)$, then so is the next set, and thus by induction the evolving set process will only explore Hamming balls.

\begin{lemma} \label{l:level}
Suppose $\eps \leq 1/2$.
For any $r \geq 0$, any $x,y\in \domain k^d$, if $\abs{x} \leq \abs{y}$, then
\[w(x, B(r)) \geq w(y, B(r)).\]
(It follows that if $S$ is a Hamming ball, then $y \in \tilde S$ implies that $x \in \tilde S$, and thus $\tilde S$ is also a Hamming ball.)
\end{lemma}

\begin{proof}
Note that $w(x,z)$ depends only on the Hamming distance $\abs{x-z}$
(coordinate-wise subtraction modulo $k$).
We first show via a symmetry argument that
\begin{equation} \label{eq:same-weig}
  w(x,B(r)) = w(y,B(r)) \text{~whenever $\abs{x} = \abs{y}$} .
\end{equation}
To this end, we will construct a permutation $\pi$ on $\domain k^d$ that 
(i) preserves Hamming distances: $\abs{\pi(a) - \pi(b)} = \abs{a - b}$ 
 for all $a,b\in \domain k^d$, 
(ii) $\pi(x) = y$, and 
(iii) $\pi(0^d) = 0^d$.
Assuming this permutation exists, we get that $z\in B(r)$ if and only if
$\pi(z)\in B(r)$, since $\abs{\pi(z) - 0} = \abs{\pi(z) - \pi(0)} = \abs{z -
0}$.
Also, $\abs{x - z} = \abs{\pi(x) - \pi(z)} = \abs{y - \pi(z)}$, thus
\[ w(x,z) = w(y,\pi(z)) . \]
Summing this equality over all $z\in B(r)$, we get (\ref{eq:same-weig}).

We now construct such a permutation $\pi$.
Take any bijection $\sigma$ on $\domain d$ that maps $I := \set{i\mid x_i
\neq 0}$ onto $\set{i\mid y_i \neq 0}$.
For $i\in I$, let $\tau_i$ be the permutation on $\domain k$ that simply
swaps $x_i$ and $y_{\sigma(i)}$.
Then we define $\pi(a) = b$ where $b_i = a_i$ if $i\notin I$, and $b_i =
\tau_i(a_{\sigma(i)})$ otherwise.
It is easy to verify that $\pi$ has all the required properties.

We now deal with the general case $\abs x < \abs y$.
It suffices to prove the lemma assuming $\abs y = \abs x + 1$.
By (\ref{eq:same-weig}), we may assume that $x$ is the indicator vector on a
subset $S = \domain c$ for some $c$, and $y$ is the indicator vector on
$\domain{c+1}$, so that $y$ differs from $x$ only at position $c$.
Picking a random neighbor $Z$ of $x$ is equivalent to picking $W = x-z$ and
setting $Z = x+W$, so the lemma is equivalent to $\Pr_W[x+W \in B(r)]
\geq \Pr_W[y+W \in B(r)]$.
In fact, we will show this inequality conditioned on all values of $W_i$
except $i = c$.
Let $w_{-c}\in \domain k^{\domain d\setminus \set c}$ be a fixing of all
those values.
We will show
\begin{equation}\label{eq:cond}
  \Pr_{W_c}[x+W \in B(r) \mid W_{-c} = w_{-c}] \geq \Pr_{W_c} [y+W \in B(r)
  \mid W_{-c} = w_{-c}] .
\end{equation}
There are two cases.
If $\abs{x_{-c} + w_{-c}} \neq r$, then $x+W$ and $y+W$ are both in $B(r)$ or
both outside of $B(r)$, and therefore (\ref{eq:cond}) holds as an equality.
In the remaining case, the left hand side of (\ref{eq:cond}) is at least
$\Pr[W_c = 0] \geq 1-\eps$,
while the right hand side is at most $\Pr[W_c \neq 0] \leq \eps$, so the
inequality follows by our assumption that $\eps \leq 1/2$.
\end{proof}

We now show that any small Hamming ball has large expansion.
The same result appears earlier in~\cite{CMN}.
We give a proof below, filling in some missing details.
The main idea is to show that the Gaussian noise graph is a small-set expander using a hypercontractiviy inequality, and to use the central limit theorems to translate this result to reason about the Hamming balls in the noisy hypercube graph.
This connection between Gaussian noise graphs and noisy hypercubes was used commonly in showing integrality gap examples for convex relaxations,
and here it is used in showing limitations for random walks based algorithms.

\begin{lemma} \label{l:hamming}
For any $\eps,\eta > 0$, there exists $\delta = \delta_{\eps,\eta}$
(independent of $k$) such that for any sufficiently large $d \geq
d_{\eps,\delta,\eta}$, all Hamming balls of size $\leq \delta n$ has
expansion $1-O(\eta)$.
\end{lemma}

\begin{proof}
We will analyze the expansion of a Hamming ball $B(r)$ by relating $B(r)$
to halfspaces $A_{r'} := \set{x\in \R \mid x\leq r'}$ in Gaussian
probability space.

Consider drawing a random edge $(x,y)$ from $H$ according to its weight,
and we would like to analyze the probability that both vertices are in $B(r)$.
The event $x\in B(r)$ is the same as $\abs x = \sum_{i \in \domain d} 
\1(x_i \neq 0) \leq r$.
Since $\abs x$ is a sum of independent random variables and each summand has
bounded third moment, by Berry--Esseen central limit theorem, for large $d$,
the sum is closely approximated by a Gaussian random variable $g$ with the
same mean and variance as $\abs x$.
That is, for all large enough $d \geq d_{\eps,\delta,\eta}$,
\begin{equation} \label{eq:clt}
  \Pr[\abs x \leq r] \approx_{\delta'} \Pr[g \leq r] \text{~for all $r\in
  \R$}
\end{equation}
for some $\delta'$ depending on $\eps,\delta,\eta$ to be
specified later.
Here we write $C \approx_{\delta'} D$ to mean $\abs{C - D} \leq \delta'$.

Moreover, multivariate central limit theorem (e.g.~\cite{Saz68}) implies that
the event $(\abs x \leq r) \wedge (\abs y \leq r)$ has roughly the same
probability as the event $(g \leq r) \wedge (h \leq r)$, 
where the bivariate Gaussian $(g,h)$ has the same
mean and covariance as $(\abs x, \abs y)$.
That is, for large enough $d$,
\begin{equation} \label{eq:bi-clt}
Pr[(\abs x \leq r) \wedge (\abs y \leq r)] \approx_{\delta'} 
Pr[(g \leq r) \wedge (h \leq r)] \text{~for all $r\in \R$}.
\end{equation}
We note that the following calculations do not depend on the dimension $d$ 
other than the CLT approximation errors 
(as we are not concerned about the graph size $n=k^d$), 
so we can choose a very large $d$ at the end 
to make the CLT approximation errors $\delta'$ to be arbitrarily small 
for the proof to go through.

Shift $g$ and $h$ to have zero mean and renormalize them to have unit
variance.
We get $g' = (g - \E[g])/\sqrt{\var[g]}$ from $g$ and similarly $h'$ from
$h$.
Then $g'$ and $h'$ have covariance
\[ \cov(g',h') = \frac{\cov(\abs x, \abs y)}{\sqrt{\var[\abs x] \var[\abs
  y]}} = \frac{\cov(\1(x_1 \neq 0), \1(y_1 \neq 0))}{\sqrt{\var[\1(x_1 \neq
0)] \var[\1(y_1 \neq 0)]}} . \]
We have $\cov(\1(x_1 \neq 0), \1(y_1 \neq 0)) = (1-\eps)\var[\1(x_1 \neq
0)]$, so $\cov(g',h') = 1-\eps$.
Therefore,
\[ \Pr[(g \leq r) \wedge (h \leq r)] = 
   \Pr[(g' \leq r') \wedge (h' \leq r')] , \]
where $r'$ is chosen so that $\Pr[g' \leq r'] = \Pr[g \leq r]$.

For standard Gaussians $g'$ and $h'$ with covariance $1-\eps$, we claim that
\begin{equation}\label{eq:gauss-stable}
  \Pr[h' \leq r' \mid g' \leq r'] \leq \Pr[g' \leq r']^{\eps/2} .
\end{equation}
This inequality follows from Gaussian hypercontractive inequality
(e.g.~\cite[Section~11.1]{O14})
\[ \E_{\text{$(1-\eps)$ correlated $g',h'$}}[\1_{\leq r'}(g')\1_{\leq
r'}(h')] \leq \norm{\1_{\leq r'}}_{2-\eps}^2 , \]
where $\1_{\leq r'}(g') := \1(g' \leq r')$ is the indicator function for the
halfspace $A_{r'}$, and the fact that
\[ \norm{\1_{\leq r'}}_{2-\eps} = \E[\1_{\leq r'}(g')^{2-\eps}]^{1/(2-\eps)}
= \E[\1_{\leq r'}(g')]^{1/(2-\eps)} = \Pr[g' \leq r']^{1/(2-\eps)} . \]

Set $\delta := \eta^{2/\eps}$.
Note that for $r'$ small enough so that $\Pr[g' \leq r'] \leq \delta$,
then (\ref{eq:gauss-stable}) implies that the halfspace 
$A_{r'}$ has expansion 
\[\Pr[h' > r \mid g' \leq r'] = 1 - \Pr[h' \leq r' \mid g' \leq r'] \geq 1 -
\Pr[g' \leq r']^{\eps/2} \geq 1 - \delta^{\eps/2} = 1-\eta.\]
We use (\ref{eq:clt}) and (\ref{eq:bi-clt}) to translate this expansion 
result from Gaussian space to the noisy hypercube.
Let $\delta'' \leq \delta$ be a constant depending on $\eps,\delta,\eta$ to
be specified later.
Any Hamming ball of $H$ of size at least $\delta'' n$ corresponds to a
halfspace of roughly the same Gaussian measure via (\ref{eq:clt}), and has
roughly the same noise stability via (\ref{eq:bi-clt}).
Choosing the CLT approximation error $\delta' := \delta''\eta$, we can
ensure that all Hamming balls $B(r)$ of size between $\delta'' n$ and
$\delta n$ have expansion $\geq 1-\eta$.
Indeed,
\[ \Pr[\text{$x$ and $y \in B(r)$}] \leq \Pr[\text{$g'$ and $h' \in A_{r'}$}]
+ \delta' \]
by (\ref{eq:bi-clt}) and
\[ \delta \geq \Pr[x \in B(r)] \geq \Pr[g' \in A_{r'}]/2 \]
by (\ref{eq:clt}) and our assumption that $\Pr[x\in B(r)] \geq \delta''$, so
\[ 1 - \phi(B(r)) = \frac{\Pr[\text{$x$ and $y \in B(r)$}]}{\Pr[x\in B(r)]}
\leq 2 \frac{\Pr[\text{$g'$ and $h'\in A_{r'}$}] + \delta'}{\Pr[g'\in
A_{r'}]} . \]
We now analyze the right hand side.
We have
\[ \frac{\Pr[\text{$g'$ and $h'\in A_{r'}$}]}{\Pr[g'\in A_{r'}]} \leq
\Pr[g'\in A_{r'}]^{\eps/2} \leq O(\eta) \]
by (\ref{eq:gauss-stable}) and the fact that $\Pr[g'\in A_{r'}] \leq 2\Pr[x\in
B(r)] \leq 2\delta$.
Also $\delta'/\Pr[g'\in A_{r'}] \leq 2\delta'/\Pr[x\in B(r)] \leq 2\delta'/\delta'' = 2\eta$.
Therefore $\phi(B(r)) \geq 1-O(\eta)$, for those $B(r)$ of size between $\delta'' n$ and $\delta n$, as required.

To deal with Hamming balls of size smaller than $\delta'' n$, 
we simply apply the hypercontractive inequality on $H$ directly.
For any subset $B$ on $H$ (not necessarily a Hamming ball), we have
\[ \Pr[y\in B \mid x\in B] \leq \Pr[x\leq B]^{c\eps/\log(1/\delta)} \]
for some $c > 0$.
The exponent $c\eps/\log(1/\delta)$ is from~\cite{W}.
Taking $\delta'' := \eta^{\log(1/\delta)/c\eps}$, we see that $\phi(B) \geq 1-\eta$ whenever $|B| \leq \delta'' n$.
\end{proof}

The key point of Lemma~\ref{l:hamming} is that everything is independent of $k$.
Therefore, given any $\eps$, we just need to set $k \geq 1/\delta$ so that $H$ has a set of expansion $\eps$ and size $\leq \delta n$ (by Claim~\ref{c:dimension}),
while the evolving set process only explores Hamming balls 
(by Lemma~\ref{l:level})
and all Hamming balls of size $\leq \delta_{\eps,\eta} n$ have expansion $\geq
1 - \eps$ (by Lemma~\ref{l:hamming}).
This proves Theorem~\ref{t:hypercube} that the evolving set process fails on the $k$-ary $\eps$-noisy hypercube with probability one.

\subsection*{Random Walks, Personal Pagerank, and Heat Kernels}

The random walk local graph partitioning algorithm~\cite{ST,ABS,KL12} works by computing the vector $p_t := A^t \chi_v$ for every vertex $v$ for $1 \leq t \leq O(\log n)$, sorting the vertices so that $p_t(1) \geq p_t(2) \geq \ldots \geq p_t(n)$, and trying all the level sets $\{1, \ldots, j\}$ for $1 \leq j \leq n$.
Using the same $k$-ary noisy hypercube example,
it is not difficult to see from Lemma~\ref{l:level} that all the level sets that the algorithm explored are Hamming balls, 
and thus the random walk algorithm will also fail to disprove the small-set expanson hypothesis.

The same argument also applies to the personal pagerank algorithm~\cite{ACL,ZLM} and the heat kernel algorithm, which work by computing some related vectors and trying all the level sets.
We note that the vectors used by these algorithms are just convex combinations of the random walk vectors $A^t \chi_v$ for different $t$, and therefore all the level sets are Hamming balls, and hence these algorithms also fail for the same reason.

We believe that this example exposes the limitations of all known local graph partitioning algorithms, and can be used as a basis to prove further lower bounds.
An interesting question is to study whether the analysis of the $O(\sqrt{\phi(S) \log(|S|)})$-approximation of the evolving set algorithm in Theorem~\ref{t:approx} is tight when $\eps = 1/\log(|S|)$.

\section*{Acknowledgement}

This research started while Tsz Chiu and Lap Chi were long-term participants in the Algorithmic Spectral Graph Theory program at the Simons Institute for the Theory of Computing in Fall 2014.
Tsz Chiu completed this work while he was a postdoc at EPFL.
Siu On completed this work while he was a postdoc at Microsoft Research New England, and he would like to thank Lorenzo Orecchia for helpful discussions.
Lap Chi completed this work while he was a visiting researcher in UC Berkeley, and he would like to thank Luca Trevisan for financial support through the NSF Grant 1216642.
We thank Shayan Oveis Gharan for comments that improved the presentation of this paper.

\bibliographystyle{plain}

\end{document}